\documentclass[11pt]{article}

\usepackage[margin=1in]{geometry}
\usepackage[utf8]{inputenc}
\usepackage[T1]{fontenc}
\usepackage{amsmath,amssymb,amsthm,mathtools}
\usepackage{amsfonts}
\usepackage{algorithm}
\usepackage{algpseudocode}
\usepackage{float}
\usepackage{booktabs}
\usepackage{graphicx}
\usepackage{longtable}
\usepackage[numbers]{natbib}
\usepackage{xcolor}
\usepackage{url}
\usepackage[colorlinks=true,citecolor=blue,linkcolor=blue,urlcolor=blue,hypertexnames=false]{hyperref}
\usepackage{enumitem}

\hypersetup{
    pdftitle={Approximating the Rural Postman Problem below 3/2},
    pdfauthor={Hong Li, Jianping Li, Wei Li, Runtao Xie, Xiaoxiao Yang},
    pdfkeywords={arc routing, rural postman problem, approximation algorithms, maximum entropy distributions}
}

\allowdisplaybreaks[1]
\newtheorem{theorem}{Theorem}[section]
\newtheorem{lemma}[theorem]{Lemma}

\newtheorem{corollary}[theorem]{Corollary}

\theoremstyle{definition}
\newtheorem{definition}[theorem]{Definition}
\theoremstyle{remark}
\newtheorem{remark}[theorem]{Remark}

\title{Approximating the Rural Postman Problem below 3/2%
\thanks{This paper is supported by the Natural Science Foundation of China [No.12361066] and Project of Yunling Scholars Training of Yunnan Province [No.K264202011820].}}

\author{Hong Li\thanks{Email: \texttt{honglimath@126.com}.}, 
Jianping Li\thanks{Corresponding author. Email: \texttt{jianping@ynu.edu.cn}.}, 
Wei Li, Runtao Xie, and Xiaoxiao Yang\\[0.5em]
School of Mathematics and Statistics, Yunnan University}
\date{}

\begin{document}
\maketitle

\begin{abstract}
We give an approximation algorithm for the rural postman problem with approximation ratio strictly smaller than $3/2$. We obtain this result by adapting to the rural postman problem the technique of sampling from maximum entropy distributions for the metric traveling salesman problem of Karlin, Klein, and Oveis Gharan. We also observe that, for every fixed $\varepsilon>0$, any $\alpha$-approximation algorithm for the metric traveling salesman problem yields an $(\alpha+\varepsilon)$-approximation algorithm for the rural postman problem; this implication is already implicit in the treatment of edges that must be traversed in the work of Lampis on the inapproximability of the traveling salesman problem.
\end{abstract}


\section{Introduction}\label{sec:intro}

\subsection{The model}\label{sec:model}

The rural postman problem (RPP) is the edge-required special case of the general routing problem (GRP) introduced by \citet{Orloff1974}. In RPP, we are given a connected undirected graph $G=(V,E)$, a nonnegative edge cost function $c:E\to\mathbb{Q}_{\ge0}$, and a set $R\subseteq E$ of required edges. A feasible solution is a closed walk in $G$ that traverses every edge in $R$ at least once. The cost of a walk is the sum of the costs of its traversed edges, counted with multiplicity. The goal is to find a feasible solution with minimum cost.

\subsection{Related work}\label{sec:related}

The metric traveling salesman problem (metric TSP) is another classical routing problem. Given a complete undirected graph with nonnegative edge costs satisfying the triangle inequality, metric TSP asks for a minimum-cost tour visiting all vertices. The classical algorithm of \citet{Christofides1976} gives a $3/2$-approximation for metric TSP. More recently, \citet{KarlinKleinOveisGharanSTOC2021,KarlinKleinOveisGharanOR2024} gave a randomized approximation algorithm with ratio strictly smaller than $3/2$ by using the technique of sampling from maximum entropy distributions. Their later work proved the LP-relative version of this guarantee for the subtour elimination LP of metric TSP \citep{KarlinKleinOveisGharanFOCS2022}. \citet{KarlinKleinOveisGharanIPCO2023} subsequently derandomized the resulting algorithm. For further improvements to the approximation ratio for metric TSP, see \citep{GurvitsKleinLeakeICALP2024,Song2026SubtourGap}. 

Metric TSP can be reduced to RPP, so RPP is at least as hard as
metric TSP \citep{vanBevernEtAl2015}. Conversely, for exact
optimization, RPP can be transformed in polynomial time into TSP
\citep{JungerReineltRinaldi1995}; thus the two problems are
polynomially interreducible in this sense. These exact transformations,
however, do not by themselves preserve multiplicative approximation
guarantees, and hence do not imply an approximation equivalence between
RPP and metric TSP.

\citet{Frederickson1979} observed that the Christofides approach can be adapted to give a $3/2$-approximation for RPP. Even after the better-than-$3/2$ result for metric TSP of \citet{KarlinKleinOveisGharanSTOC2021,KarlinKleinOveisGharanOR2024}, recent treatments continue to record $3/2$ as the approximation guarantee for RPP; see, for example, \citet{BoyaciDangLetchford2023} and \citet{AgarwalAkella2023}.

\subsection{Our result}\label{sec:contribution}

We show that the technique of sampling from maximum entropy distributions for metric TSP \citep{KarlinKleinOveisGharanSTOC2021,KarlinKleinOveisGharanOR2024} can be adapted to RPP, yielding an approximation algorithm with approximation ratio strictly below $3/2$. The main result is the following.

\begin{theorem}\label{thm:main}
There is a constant $\epsilon > 10^{-36}$ and a polynomial-time randomized $(3/2-\epsilon)$-approximation algorithm for RPP.
\end{theorem}

We state Theorem~\ref{thm:main} in the randomized form because it
admits a direct proof using the quantitative maximum-entropy slack
theorem quoted below. A stronger LP-relative conclusion is also
available from the subsequent analysis of
\citet{KarlinKleinOveisGharanFOCS2022}, and a corresponding
LP-relative better-than-$3/2$ guarantee can be attained
deterministically using the conditional-expectation derandomization
of \citet{KarlinKleinOveisGharanIPCO2023}; see
Remark~\ref{rem:lp-relative-derandomization}. We record these stronger
consequences separately because their verification in the present
setting requires combining the LP-relative analysis with the
derandomization framework of those works, rather than invoking a
single stated result that applies verbatim to the fixed-required-edge
formulation used here.

The randomized algorithm first applies a preprocessing similar to the transformations used by \citet{ChristofidesEtAl1981} and \citet{JungerReineltRinaldi1995}. On the resulting instance, we adapt the maximum-entropy approach for metric TSP: we sample a spanning tree while ensuring that all required edges are retained, perform a minimum-cost parity correction to obtain a connected Eulerian multigraph, and then expand it into a feasible RPP closed walk. The analysis bounds the expected cost of the sampled tree together with the cost of the parity correction.

The same preprocessing also makes explicit a general approximation transfer that is already implicit in the treatment of edges that must be traversed in \citet[Section~2.1]{Lampis2014}.

\begin{theorem}\label{thm:generic-transfer}
Let $\alpha\ge1$. If metric TSP admits a polynomial-time $\alpha$-approximation algorithm, then, for every fixed $\varepsilon>0$, RPP admits a polynomial-time $(\alpha+\varepsilon)$-approximation algorithm.
\end{theorem}

Before the better-than-$3/2$ approximation algorithms for metric TSP, this implication did not improve the classical $3/2$ guarantee for RPP. Moreover, transferring the currently known improvement below $3/2$ through this reduction requires solving a metric TSP instance of extremely large size. Thus, although the reduction is polynomial time for every fixed $\varepsilon>0$, its better-than-$3/2$ consequence is mainly of theoretical interest. The maximum-entropy algorithm analyzed in this paper avoids this increase in instance size.

After the preprocessing, we formulate an LP relaxation for RPP, denoted by \textup{\textsc{RPP-LP}}. For completeness, Appendix~\ref{app:lp-consequences} records two further consequences of this LP relaxation on the preprocessed RPP instances. The approximation transfer above also holds relative to \textup{\textsc{RPP-LP}}, and the worst-case integrality gap of \textup{\textsc{RPP-LP}} is equal to that of the subtour elimination LP for metric TSP; in particular, it is strictly below $3/2$ by \citet{KarlinKleinOveisGharanFOCS2022}.

\subsection{Organization}
Section~\ref{sec:prelim} fixes notation and records the polyhedral tools used later. Section~\ref{sec:preprocessing} gives the preprocessing and the LP relaxation for the preprocessed RPP instance. Section~\ref{sec:maxent} explains maximum-entropy distributions. Section~\ref{sec:algorithm} states the maximum-entropy algorithm and proves Theorem~\ref{thm:main}. Section~\ref{sec:generic-reduction} makes explicit the approximation transfer in Theorem~\ref{thm:generic-transfer}, and Section~\ref{sec:conclusion} concludes the paper. Appendix~\ref{app:lp-consequences} records the LP-relative transfer and the integrality-gap equality.

\section{Preliminaries}\label{sec:prelim}

\subsection{Notation}\label{subsec:notation}

Let $G=(V,E)$ be a graph. For $S\subseteq V$, let $\delta_G(S)$ denote the set of edges with exactly one endpoint in $S$, and let $E_G(S)$ denote the set of edges with both endpoints in $S$. We write $\delta_G(v):=\delta_G(\{v\})$. Within a statement or proof in which a single ambient graph $G$ has been fixed, we suppress the subscript and write $\delta(S)$ and $E(S)$.

For a vector $x\in\mathbb R^E$ and an edge set $A\subseteq E$, write $x(A):=\sum_{e\in A}x_e$. For an edge set $A$, let $\chi_A$ denote its incidence vector; when $C$ is a tour, $\chi_C$ denotes the incidence vector of its edge set. For edge multisets $A$ and $B$, write $A\uplus B$ for their multiset union. For an edge multiset $A$, let $\operatorname{odd}(A)$ denote the set of vertices of odd degree in the corresponding multigraph.

For an edge cost function $c$ and an edge multiset $A$, let $c(A)$ denote the sum of the costs of the edges in $A$, counted with multiplicity. For a vector $x\in\mathbb R^E$, write $c(x):=\sum_{e\in E}c_ex_e$. If $W$ is a walk, we write $c(W)$ for the cost of its edge multiset.

For an RPP instance $J$ and a metric TSP instance $I$, let $\operatorname{OPT}_{\mathrm{RPP}}(J)$ and $\operatorname{OPT}_{\mathrm{TSP}}(I)$ denote their optimum values. When the instance is clear from the context, we omit the argument.

\subsection{The subtour elimination LP, spanning-tree polytope, and \texorpdfstring{$Q$}{Q}-join polyhedron}\label{subsec:polyhedra}

We first recall the subtour elimination LP for metric TSP \citep{DantzigFulkersonJohnson1954}. Let $(V,\binom{V}{2})$ be a complete graph with a nonnegative metric cost function $c$. The subtour elimination LP is
\begin{align}
\min \quad & \sum_{e\in\binom{V}{2}} c_e x_e \nonumber\\
\text{s.t.}\quad
& x(\delta(v))=2,
&& \forall\,v\in V, \label{eq:subtour-degree}\\
& x(\delta(S))\ge2,
&& \forall\,\emptyset\ne S\subsetneq V, \label{eq:subtour-cut}\\
& x_e\ge0,
&& \forall\,e\in\binom{V}{2}. \label{eq:subtour-nonneg}
\end{align}
We refer to this relaxation as \textup{\textsc{TSP-LP}}. When these constraints are applied to a graph $G=(V,E)$ that is not complete, a vector $x\in\mathbb R^E$ is extended to $\binom{V}{2}$ by setting $x_e=0$ for every $e\in\binom{V}{2}\setminus E$. We denote the optimum value of \textup{\textsc{TSP-LP}} on a metric TSP instance $I$ by $L_{\mathrm{TSP}}(I)$.

For a connected graph $G=(V,E)$, the spanning-tree polytope is the set of vectors $x\in\mathbb R^E$ satisfying the following inequalities of \citet{Edmonds1970}:
\begin{align}
x(E)&=|V|-1, \label{eq:tree-total}\\
x(E(S))&\le |S|-1,
&& \forall\,\emptyset\ne S\subsetneq V, \label{eq:tree-rank}\\
x_e&\ge0,
&& \forall\,e\in E. \label{eq:tree-nonneg}
\end{align}
For any edge cost function $c:E\to\mathbb R$, the minimum cost of a spanning tree equals the minimum of $\sum_{e\in E}c_ex_e$ over this polytope. Hence a minimum-cost spanning tree can be computed in polynomial time, and its cost is at most $\sum_{e\in E}c_ex_e$ for every vector $x$ in the spanning-tree polytope.

Let $Q\subseteq V$ have even cardinality. A $Q$-join is an edge multiset $A$ such that $\operatorname{odd}(A)=Q$. In particular, if $A$ is a connected edge multiset, then adding an $\operatorname{odd}(A)$-join makes every vertex have even degree. The $Q$-join polyhedron of Edmonds and Johnson \citep{EdmondsJohnson1973} is
\begin{align}
x(\delta(S))&\ge1,
&& \forall\,S\subseteq V\text{ with }|S\cap Q|\text{ odd}, \label{eq:qjoin-cut}\\
x_e&\ge0,
&& \forall\,e\in E. \label{eq:qjoin-nonneg}
\end{align}
The minimum cost of a $Q$-join equals the minimum of $\sum_{e\in E}c_ex_e$ over this polyhedron. Hence a minimum-cost $Q$-join can be computed in polynomial time, and its cost is at most $\sum_{e\in E}c_ex_e$ for every vector $x$ in the $Q$-join polyhedron.

\section{Preprocessing and the RPP linear program}\label{sec:preprocessing}

In this section, we transform the input RPP instance into a complete graph in which the required edges form a matching and all other edges correspond to shortest paths in the original graph. We then formulate an LP relaxation that preserves the required-edge structure used in Section~\ref{sec:algorithm}. The preprocessed graph also forms the starting point for Section~\ref{sec:generic-reduction}.

Throughout this section, let $G=(V,E)$ be the original connected graph with nonnegative edge costs $c$, and let $R=\{e_1,\ldots,e_m\}$ be the set of required edges. The cases $m=0$ and $m=1$ are handled directly by the algorithm in Section~\ref{sec:algorithm}; hence the discussion below assumes $m\ge2$.

\subsection{The preprocessed complete graph}\label{subsec:preprocessed-graph}

We use a preprocessing similar to the transformations of \citet{ChristofidesEtAl1981} and \citet{JungerReineltRinaldi1995}. For each required edge $e_i=a_ib_i\in R$, create two private vertices $i^-$ and $i^+$. Let $\widehat V=\{i^-,i^+:i=1,\ldots,m\}$, and define the projection map $\pi:\widehat V\to V$ by $\pi(i^-)=a_i$ and $\pi(i^+)=b_i$. For each $i$, add the edge $\widehat e_i=i^-i^+$ with cost $\widehat c_{\widehat e_i}=c_{e_i}$. These edges are the required edges of the preprocessed graph, and we denote their set by $\widehat R=\{\widehat e_1,\ldots,\widehat e_m\}$.

For every edge $uv\in\binom{\widehat V}{2}\setminus\widehat R$, set $\widehat c_{uv}=\operatorname{dist}_G(\pi(u),\pi(v))$, where $\operatorname{dist}_G$ denotes the shortest-path distance with respect to $c$. We call these edges connector edges and denote their set by $\widehat E_{\mathrm{con}}=\binom{\widehat V}{2}\setminus\widehat R$. Thus $\widehat G=(\widehat V,\widehat E)$, where $\widehat E=\widehat R\cup\widehat E_{\mathrm{con}}$, is a complete graph whose edges are partitioned into the required edges $\widehat R$ and the connector edges $\widehat E_{\mathrm{con}}$.

The preprocessed graph has the following properties.
\begin{enumerate}[label=\textup{(P\arabic*)}]
\item The required edges in $\widehat R$ form a matching, and every vertex of $\widehat G$ is incident to exactly one edge of $\widehat R$.
\item The cost of a connector edge $uv$ is the shortest-path distance from $\pi(u)$ to $\pi(v)$ in the original graph. The full cost function $\widehat c$ need not be metric, because the cost of a required edge $\widehat e_i$ is $c_{e_i}$, which may exceed $\operatorname{dist}_G(a_i,b_i)$.
\item A tour in $\widehat G$ that visits every vertex exactly once and traverses every edge of $\widehat R$ can be expanded into a feasible RPP closed walk in $G$ of the same cost. Conversely, every feasible RPP closed walk can be transformed into such a tour in $\widehat G$ without increasing its cost.
\end{enumerate}

The next lemma makes the last property precise.

\begin{lemma}\label{lem:preprocess-opt}
The minimum cost of a tour in $\widehat G$ that visits every vertex exactly once and traverses every edge of $\widehat R$ is equal to the optimum value $\operatorname{OPT}_{\mathrm{RPP}}$ of the original RPP instance on $G$ with required edge set $R$.
\end{lemma}

\begin{proof}
Let $W$ be a feasible RPP closed walk in the original graph. For each required edge $e_i\in R$, choose one traversal of $e_i$ in $W$, and read the chosen traversals in their cyclic order along $W$. A chosen traversal of $e_i$ determines whether the corresponding required edge $\widehat e_i$ is traversed from $i^-$ to $i^+$ or from $i^+$ to $i^-$. Between two consecutive chosen traversals, $W$ contains a walk joining the corresponding endpoints in the original graph. Replace this walk by the connector edge between the corresponding private vertices. By the definition of its cost, this replacement does not increase the total cost. Since every required edge is chosen exactly once and all private vertices are distinct, the resulting tour visits every vertex of $\widehat G$ exactly once and traverses every edge of $\widehat R$. Hence its cost is at most $c(W)$, and the minimum cost of such a tour is at most $\operatorname{OPT}_{\mathrm{RPP}}$.

Conversely, let $C$ be a tour in $\widehat G$ that visits every vertex exactly once and traverses every edge of $\widehat R$. Replace every required edge $\widehat e_i$ by the corresponding original required edge $e_i$, and replace every connector edge $uv$ by a shortest path in $G$ from $\pi(u)$ to $\pi(v)$. The resulting closed walk traverses every edge of $R$ and is therefore feasible for the RPP. Its cost is equal to $\widehat c(C)$. Hence $\operatorname{OPT}_{\mathrm{RPP}}$ is at most the minimum cost of such a tour in $\widehat G$.
\end{proof}

\subsection{The RPP linear program}\label{subsec:rpp-lp}

On the preprocessed graph $\widehat G$ with required edge set $\widehat R$, we use the following subtour-elimination relaxation for finding a minimum-cost tour that traverses every edge of $\widehat R$. The variables are $x_e$ for the connector edges $e\in\widehat E_{\mathrm{con}}$. For a vector $x\in\mathbb R^{\widehat E_{\mathrm{con}}}$, let $\overline x\in\mathbb R^{\widehat E}$ be its extension defined by $\overline x_e=x_e$ for $e\in\widehat E_{\mathrm{con}}$ and $\overline x_e=1$ for $e\in\widehat R$. Thus every required edge in $\widehat R$ is fixed to one rather than represented by a variable. The relaxation is
\begin{align}
\min \quad &
\sum_{e\in\widehat E_{\mathrm{con}}}\widehat c_e x_e
+\sum_{e\in\widehat R}\widehat c_e
\nonumber\\
\text{s.t.}\quad
& \overline x(\delta(v))=2,
&& \forall\, v\in\widehat V, \label{eq:rpp-lp-degree}\\
& \overline x(\delta(S))\ge2,
&& \forall\, \emptyset\ne S\subsetneq\widehat V, \label{eq:rpp-lp-cut}\\
& x_e\ge0,
&& \forall\, e\in\widehat E_{\mathrm{con}}. \label{eq:rpp-lp-nonneg}
\end{align}
We call this relaxation \textup{\textsc{RPP-LP}}. Let $x^*$ be an optimum extreme point solution, let $\overline x^*$ be its extension to $\widehat E$, and let $L$
denote the optimum value.

\begin{lemma}\label{lem:rpp-lp-lower}
The optimum value $L$ of \textup{\textsc{RPP-LP}} can be computed in polynomial time and satisfies $L\le \operatorname{OPT}_{\mathrm{RPP}}$, where $\operatorname{OPT}_{\mathrm{RPP}}$ refers to the original RPP instance on $G$ with required edge set $R$.
\end{lemma}

\begin{proof}
The degree constraints and nonnegativity constraints are explicit. Given a vector $x$, a violated cut constraint can be found, if one exists, by computing a global minimum cut in $\widehat G$ with capacities $\overline x_e$. Hence \textup{\textsc{RPP-LP}} can be optimized in polynomial time by
the ellipsoid method \citep{GroetschelLovaszSchrijver1988}; after optimizing the objective, polynomially many additional
lexicographic optimization calls over the optimal face yield an optimum
extreme point solution $x^*$ in polynomial time. The extreme-point property is used later to justify the polynomial-time implementation of the approximate maximum-entropy sampling step; see the proof of Lemma~\ref{lem:contract-required}.

Let $C$ be a minimum-cost tour in $\widehat G$ that visits every vertex exactly once and traverses every edge of $\widehat R$. Set $x_e=1$ for every connector edge $e\in\widehat E_{\mathrm{con}}$ used by $C$, and $x_e=0$ for every other connector edge. Then $\overline x$ is the incidence vector of $C$, and hence $x$ is feasible for \textup{\textsc{RPP-LP}}. Its objective value is $\widehat c(C)$. By Lemma~\ref{lem:preprocess-opt}, $\widehat c(C)=\operatorname{OPT}_{\mathrm{RPP}}$ for the original RPP instance on $G$ with required edge set $R$. Therefore $L\le\operatorname{OPT}_{\mathrm{RPP}}$.
\end{proof}

\section{Maximum-entropy distributions}\label{sec:maxent}

\subsection{From the RPP LP to the spanning-tree polytope}\label{subsec:split-edge}

We use the following standard fact about \textup{\textsc{TSP-LP}}. If $x^0$ is feasible for \textup{\textsc{TSP-LP}} and $x^0_{e_0}=1$ for some edge $e_0$, then the restriction of $x^0$ obtained by deleting $e_0$ is a vector in the spanning-tree polytope. A convenient way to create such an edge is to split a vertex and introduce a zero-cost edge between the two copies, as is also done in the metric TSP algorithm of \citet{KarlinKleinOveisGharanSTOC2021,KarlinKleinOveisGharanOR2024}.

We apply a similar split to the optimum extreme point solution $x^*$ of \textup{\textsc{RPP-LP}}. Every required edge has coordinate one in its extension $\overline x^*$, but these edges must remain in the spanning trees sampled later. We therefore create an additional zero-cost edge of coordinate one by splitting one endpoint of a required edge. Choose the required edge $\widehat e_1=ab$, where $a=1^-$ and $b=1^+$. Replace $a$ by two vertices $a^r$ and $a^c$, make the required edge $\widehat e_1$ incident to $a^r$, make every connector edge formerly incident to $a$ incident to $a^c$, and add the zero-cost edge $e_0=a^ra^c$. We keep the notation $\widehat e_1$ and $\widehat R$ for the corresponding required edge and required-edge set after the split. Let $G^+=(V^+,E^+)$ be the resulting graph and let $c^+$ be its edge cost function.

Transfer the coordinates of $\overline x^*$ to the corresponding edges of $G^+$ and set $r^0_{e_0}=1$. Let $r^0\in\mathbb R^{E^+}$ be the resulting vector, and let $r\in\mathbb R^{E^+\setminus\{e_0\}}$ be its restriction to $E^+\setminus\{e_0\}$.

\begin{lemma}\label{lem:split-vector}
The vector $r^0$, extended by zero to the edges not contained in $G^+$, is feasible for \textup{\textsc{TSP-LP}} on the complete graph with vertex set $V^+$. Moreover, $r$ is a vector in the spanning-tree polytope of $(V^+,E^+\setminus\{e_0\})$. Every required edge has coordinate one in $r$, and $c^+(r)=L$.
\end{lemma}

\begin{proof}
The degree equations follow directly from the construction. The required edge incident to $a^r$ and the edge $e_0$ both have coordinate one. The connector edges incident to $a^c$ have total coordinate one by \eqref{eq:rpp-lp-degree}, and all other degree equations are unchanged.

For the cut constraints, a cut that does not separate $a^r$ and $a^c$ corresponds directly to a cut of $\widehat G$. Suppose that a cut separates them and, after taking its complement if necessary, that $a^r$ is inside the cut and $a^c$ is outside. If $b$ is also outside, then $e_0$ and $\widehat e_1$, both of coordinate one, already give cut value two. Otherwise, let $U\subseteq\widehat V\setminus\{a\}$ be the vertices on the same side as $a^r$. The cut $\delta_{\widehat G}(U)$ contains $\widehat e_1$, of coordinate one, and its remaining edges correspond exactly to the edges of the split cut other than $e_0$. By \eqref{eq:rpp-lp-cut}, these remaining edges have total coordinate at least one. Hence every nontrivial cut has $r^0$-value at least two, and $r^0$ is feasible for \textup{\textsc{TSP-LP}}.

Since $r^0_{e_0}=1$, the degree equations give $r(E^+\setminus\{e_0\})=|V^+|-1$. For every nonempty proper $S\subsetneq V^+$,
\[
r(E_{(V^+,E^+\setminus\{e_0\})}(S))
\le r^0(E_{G^+}(S))
=|S|-\frac12r^0(\delta_{G^+}(S))
\le |S|-1.
\]
Together with nonnegativity, these are the spanning-tree polytope inequalities. The required-edge coordinates are unchanged, and $e_0$ has cost zero, so $c^+(r)=L$.
\end{proof}

\subsection{Sampling from maximum entropy distributions}\label{subsec:maxent-sampling}

We first define the spanning-tree distribution used in the analysis.

\begin{definition}[Maximum entropy distribution]\label{def:maxent}
Let $H=(W,F)$ be a connected graph, let $\mathcal T(H)$ denote the family of spanning trees of $H$, and let $q\in\mathbb R^F$ be a vector in the spanning-tree polytope of $H$. A probability distribution $\mu$ on $\mathcal T(H)$ has marginals $q$ if $\Pr_{T\sim\mu}[e\in T]=q_e$ for every $e\in F$. The entropy of $\mu$ is $\operatorname{Ent}(\mu):=-\sum_{T\in\mathcal T(H)}\mu(T)\log\mu(T)$, where $0\log0:=0$. The maximum entropy distribution with marginals $q$, denoted by $\mu_q$, is the distribution that maximizes $\operatorname{Ent}(\mu)$ over all spanning-tree distributions with marginals $q$.
\end{definition}

Since $q$ is in the spanning-tree polytope, a spanning-tree distribution with marginals $q$ exists. The strict concavity of entropy implies that $\mu_q$ is unique. If $\mu$ and $\nu$ are two distributions on the same family $\mathcal T(H)$ of spanning trees, define their total variation distance by $\operatorname{TV}(\mu,\nu):=\frac12\sum_{T\in\mathcal T(H)}|\mu(T)-\nu(T)|$.

The improvement from the maximum entropy distribution comes from parity correction. If $Q=\operatorname{odd}(T)$, a cut $S$ imposes a constraint in the $Q$-join polyhedron only when $|T\cap\delta(S)|$ is odd. The analysis of \citet{KarlinKleinOveisGharanSTOC2021,KarlinKleinOveisGharanOR2024} exploits the probabilistic structure of the maximum entropy distribution to obtain slack in these parity-correction constraints and thereby reduce the expected $Q$-join cost below the classical $1/2$ bound. The technical theorem used for this purpose is stated in Section~\ref{sec:algorithm} in the form needed here.

All spanning-tree statements in this subsection also apply to loopless multigraphs, with parallel edges treated as distinct.

The exact distribution $\mu_q$ is used in the approximation analysis. For the polynomial-time implementation, we use a weighted spanning-tree distribution. For a connected loopless multigraph $H=(W,F)$ and weights $\lambda\in\mathbb R_{\ge0}^F$ whose positive-weight support contains a spanning tree, the corresponding $\lambda$-uniform distribution is defined by
\[
\mu_\lambda(T)
:=
\frac{\prod_{e\in T}\lambda_e}
{\sum_{T'\in\mathcal T(H)}\prod_{e\in T'}\lambda_e}
\qquad (T\in\mathcal T(H)).
\]
The implementation in \citet{KarlinKleinOveisGharanOR2024} combines the algorithm of \citet{AsadpourGoemansMadryOveisGharanSaberi2017} for approximately prescribed spanning-tree marginals with the stability of maximum entropy distributions proved by \citet{StraszakVishnoi2019}. We use the following consequence.

\begin{lemma}[Consequence of \citet{KarlinKleinOveisGharanOR2024} and \citet{StraszakVishnoi2019}]\label{lem:approx-sampling}
Let $H=(W,F)$ be a connected loopless multigraph with $|W|\ge2$, let $q\in\mathbb Q^F$ be a vector in its spanning-tree polytope, and let $\mu_q$ be the maximum entropy distribution with marginals $q$. For any $\rho,\delta\in\mathbb Q\cap(0,1)$, one can deterministically compute nonnegative weights $\lambda\in\mathbb R_{\ge0}^F$ such that, writing $\widehat\mu_q:=\mu_\lambda$,
\[
\operatorname{TV}(\widehat\mu_q,\mu_q)\le\rho,
\qquad
\Pr_{T\sim\widehat\mu_q}[e\in T]\le(1+\delta)q_e
\quad\text{for every }e\in F.
\]
The deterministic computation of the weights and the subsequent randomized sampling of a tree from $\widehat\mu_q$ take time polynomial in the encoding size of $(H,q)$, $\log(1/\rho)$, $\log(1/\delta)$, and
\[
\log\frac1{q_{\min}},
\qquad
q_{\min}:=\min\{q_e:q_e>0\}.
\]
\end{lemma}

\begin{proof}
Edges with $q_e=0$ may first be deleted; assign them weight $\lambda_e=0$ when extending the computed weights to $F$. The remaining graph is connected because $q$ is a spanning-tree marginal vector. Apply the ellipsoid-based algorithm for approximately prescribed marginals of \citet{AsadpourGoemansMadryOveisGharanSaberi2017}, as used by \citet{KarlinKleinOveisGharanOR2024}, with accuracy parameter $\tau>0$, and let $\widetilde q$ be the marginal vector of the resulting $\lambda$-uniform distribution. Then
\[
\widetilde q_e\le(1+\tau)q_e
\qquad\text{for every }e\in F.
\]
Both $q$ and $\widetilde q$ are spanning-tree marginal vectors and hence have total value $|W|-1$. Therefore
\[
\|\widetilde q-q\|_1
=2\sum_{e:\,\widetilde q_e>q_e}(\widetilde q_e-q_e)
\le2\tau(|W|-1).
\]
Moreover, a $\lambda$-uniform distribution is the maximum entropy distribution for its own marginal vector. The stability theorem of \citet{StraszakVishnoi2019} consequently implies that $\operatorname{TV}(\mu_\lambda,\mu_q)\le\rho$ once $\tau$ is chosen sufficiently small. This theorem applies also to boundary points; for the spanning-tree polytope, the facet normals have entries in $\{-1,0,1\}$, so its quantitative bound permits such a choice with $\log(1/\tau)$ polynomial in the size of $H$, $\log(1/\rho)$, and $\log(1/\delta)$. Taking also $\tau\le\delta$ gives the stated marginal bound. The ellipsoid-based marginal algorithm has polynomial dependence on $\log(1/\tau)$, and weighted spanning trees can subsequently be sampled in randomized polynomial time, as in \citet{KarlinKleinOveisGharanOR2024}.
\end{proof}

We now apply this sampling result to the vector $r$ from Lemma~\ref{lem:split-vector}. Let $\mu_r$ denote the maximum entropy distribution with marginals $r$ on the spanning trees of $(V^+,E^+\setminus\{e_0\})$. Since $\mu_r$ has marginals $r$, $\mathbb E_{T\sim\mu_r}[c^+(T)]=c^+(r)=L$. Thus the maximum entropy distribution does not improve the expected tree cost; the gain comes from the parity correction. Since $r_e=1$ for every required edge $e\in\widehat R$, every tree in the support of $\mu_r$ contains all required edges. To preserve this property exactly in the polynomial-time sampling step, we contract the required edges before sampling and restore them afterwards.

Starting from $(V^+,E^+\setminus\{e_0\})$, contract every edge of $\widehat R$. For each edge $e=uv\in E^+\setminus(\widehat R\cup\{e_0\})$, retain $e$ as a distinct edge whose endpoints are the vertices containing $u$ and $v$ after the contraction. Thus different original edges remain distinct even if they become parallel. No loop is created, because there is no edge in $E^+\setminus(\widehat R\cup\{e_0\})$ joining the two endpoints of the same required edge. Let $G'=(V',E')$ be the resulting loopless multigraph. Each edge $e'\in E'$ corresponds to a unique edge $e\in E^+\setminus(\widehat R\cup\{e_0\})$; define $r'_{e'}:=r_e$.

\begin{lemma}[Sampling after contracting required edges]\label{lem:contract-required}
For any $\rho,\delta\in(0,1)$, one can sample from a distribution $\widehat\mu_r$ on the spanning trees of $(V^+,E^+\setminus\{e_0\})$ in randomized time polynomial in the input size, $\log(1/\rho)$, and $\log(1/\delta)$ such that:
\begin{enumerate}[label=\textup{(\roman*)}]
\item every tree in the support of $\widehat\mu_r$ contains every edge of $\widehat R$;
\item $\operatorname{TV}(\widehat\mu_r,\mu_r)\le\rho$;
\item $\Pr_{T\sim\widehat\mu_r}[e\in T]\le(1+\delta)r_e$ for every $e\in E^+\setminus\{e_0\}$.
\end{enumerate}
\end{lemma}

\begin{proof}
Let $\mathcal T_R$ be the family of spanning trees of $(V^+,E^+\setminus\{e_0\})$ that contain every edge of $\widehat R$, and let $\mathcal T'$ be the family of spanning trees of $G'$. We first describe the correspondence between these two families.

Let $T\in\mathcal T_R$. Contract the edges of $\widehat R$ and remove them from $T$. The resulting edge set is connected and contains $(|V^+|-1)-|\widehat R|=|V'|-1$ edges, and hence is a spanning tree of $G'$. Conversely, let $T'\in\mathcal T'$. Replace every edge $e'\in T'$ by its unique corresponding edge before the contraction and add all edges of $\widehat R$. The resulting edge set is connected and has $(|V'|-1)+|\widehat R|=|V^+|-1$ edges, and hence is a spanning tree of $(V^+,E^+\setminus\{e_0\})$ containing $\widehat R$. These two operations are inverse to each other. Thus contraction gives a bijection between $\mathcal T_R$ and $\mathcal T'$. The fact that parallel edges in $G'$ are kept distinct is essential for this correspondence, since each edge of $T'$ then determines a unique edge before the contraction.

Every spanning-tree distribution with marginals $r$ is supported on $\mathcal T_R$, because $r_e=1$ for every $e\in\widehat R$. In particular, contracting a tree drawn from $\mu_r$ gives a distribution on $\mathcal T'$. For every $e'\in E'$ corresponding to $e$, the marginal probability of $e'$ under this distribution is $r_e=r'_{e'}$. Hence $r'$ is a vector in the spanning-tree polytope of $G'$.

More generally, the bijection gives a one-to-one correspondence between spanning-tree distributions with marginals $r$ and spanning-tree distributions on $G'$ with marginals $r'$. Indeed, contraction preserves the marginals of all nonrequired edges, while in the reverse direction every edge of $\widehat R$ is added to every lifted tree and therefore has marginal one. Corresponding distributions assign the same probabilities to corresponding trees and consequently have the same entropy. By Definition~\ref{def:maxent}, contracting $T\sim\mu_r$ therefore gives exactly the maximum entropy distribution $\mu_{r'}$ with marginals $r'$ on $G'$, and $\mu_r$ is recovered by lifting $T'\sim\mu_{r'}$ and adding all edges of $\widehat R$.

Since $x^*$ is an extreme point of \textup{\textsc{RPP-LP}}, the standard determinant bound implies that the logarithm of the inverse minimum positive coordinate of $x^*$ is polynomially bounded in the input size. The positive coordinates of $r'$ are inherited from $x^*$, so Lemma~\ref{lem:approx-sampling} can be applied to $r'$ in polynomial time. Let $\widehat\mu_{r'}=\mu_{\lambda'}$ be the resulting $\lambda'$-uniform distribution. Define $\widehat\mu_r$ by sampling $T'\sim\widehat\mu_{r'}$, replacing each edge of $T'$ by its corresponding edge before the contraction, and adding all edges of $\widehat R$. Every tree in the support of $\widehat\mu_r$ therefore contains $\widehat R$, proving (i).

Since contraction and lifting are inverse bijections and preserve the probability assigned to corresponding trees, $\operatorname{TV}(\widehat\mu_r,\mu_r)=\operatorname{TV}(\widehat\mu_{r'},\mu_{r'})\le\rho$, which proves (ii). For every $e\in E^+\setminus(\widehat R\cup\{e_0\})$, let $e'$ be the corresponding edge of $G'$. Then
\[
\Pr_{T\sim\widehat\mu_r}[e\in T]
=
\Pr_{T'\sim\widehat\mu_{r'}}[e'\in T']
\le(1+\delta)r'_{e'}
=(1+\delta)r_e.
\]
If $e\in\widehat R$, then $\Pr_{T\sim\widehat\mu_r}[e\in T]=1=r_e$, so (iii) also holds for the required edges.
\end{proof}

\section{The algorithm and approximation analysis}\label{sec:algorithm}

We now state the algorithm obtained by adapting the technique of sampling from maximum entropy distributions to the RPP. The accuracy parameters $\rho$ and $\delta$ are absolute constants chosen sufficiently small in the proof of Theorem~\ref{thm:main}.

\begin{algorithm}[htbp]
\caption{Approximation algorithm for RPP}\label{alg:rpp}
\begin{algorithmic}[1]
\Require A connected undirected graph $G=(V,E)$, nonnegative edge costs $c$, and required edge set $R$.
\Ensure A feasible RPP closed walk.
\If{$R=\emptyset$}
    \State return the empty closed walk.
\EndIf
\If{$R=\{ab\}$}
    \State return $ab$ together with a shortest path from $b$ to $a$.
\EndIf
\State Construct the preprocessed graph $\widehat G=(\widehat V,\widehat E)$ and the required edge set $\widehat R$.
\State Solve \textup{\textsc{RPP-LP}} and let $x^*$ be an optimum extreme point solution of value $L$.
\State Perform the split of Section~\ref{subsec:split-edge}, obtaining $G^+=(V^+,E^+)$, the zero-cost edge $e_0$, and the vector $r$ in the spanning-tree polytope of $(V^+,E^+\setminus\{e_0\})$.
\State Sample a spanning tree $T\sim\widehat\mu_r$ according to Lemma~\ref{lem:contract-required}.
\State Let $Q=\operatorname{odd}(T)$ and compute a minimum-cost $Q$-join $J_T$ in $G^+$ with costs $c^+$.
\State Identify $a^r$ and $a^c$ in $T\uplus J_T$, set the multiplicity of $e_0$ to zero, replace every required edge by the corresponding original required edge, and expand every connector edge into its shortest path in $G$, obtaining a connected Eulerian multigraph containing every edge of $R$.
\State Compute an Eulerian closed walk $W$ in this multigraph and return $W$.
\end{algorithmic}
\end{algorithm}

The strict improvement in the parity-correction bound is obtained from the following technical theorem of \citet[Theorem~3.1]{KarlinKleinOveisGharanOR2024}. It gives the quantitative slack statement for a spanning tree drawn from the maximum entropy distribution defined above. We state it in the notation needed for our analysis.

\begin{lemma}[\citet{KarlinKleinOveisGharanOR2024}]\label{lem:karlin-main}
Let $\xi^0$ be a feasible solution of \textup{\textsc{TSP-LP}} on a metric complete graph, and suppose that an edge $e_0=u_0v_0$ has cost zero and $\xi^0_{e_0}=1$. Let $E(\xi^0):=\{e\ne e_0: \xi^0_e>0\}$, and let $\xi$ be the restriction of $\xi^0$ to $E(\xi^0)$. Let $C^\star$ be a globally optimum TSP tour that contains $e_0$, let $E(C^\star)$ denote its edge set, and define $\zeta=(\xi+\chi_{C^\star})/2$, where $\xi$ is extended by zero outside $E(\xi^0)$. Let $\mu_\xi$ be the maximum entropy distribution with marginals $\xi$.

For $0<\eta\le10^{-12}$ and $\beta>0$, there exist random functions $s:E(\xi^0)\cup\{e_0\}\to\mathbb R$ and $s^*:E(C^\star)\to\mathbb R_{\ge0}$, as functions of $T\sim\mu_\xi$, such that:
\begin{enumerate}[label=\textup{(\roman*)}]
\item $s_e\ge-\beta\xi_e$ for every $e\in E(\xi^0)$;
\item if $S\subseteq V$, $u_0,v_0\notin S$, $\zeta(\delta(S))\le2+\eta$, and $|T\cap\delta(S)|$ is odd, then $s(\delta(S))+s^*(\delta(S))\ge0$;
\item $\mathbb E[s^*_{e^*}]\le218\eta\beta$ for every $e^*\in E(C^\star)$, and $\mathbb E[s_e]\le-\frac13\epsilon_{\mathrm P}\beta\xi_e$ for every $e\in E(\xi^0)$, where $\epsilon_{\mathrm P}=3.12\cdot10^{-16}$ is the constant defined in Equation~(7.4) of \citet{KarlinKleinOveisGharanOR2024}.
\end{enumerate}
\end{lemma}

For the remainder of the analysis, take a minimum-cost tour in $\widehat G$ that visits every vertex exactly once and traverses every edge of $\widehat R$. At $a$, this tour uses the required edge $\widehat e_1$ and one connector edge. Splitting $a$ as in Section~\ref{subsec:split-edge} and inserting $e_0$ between these two edges gives a tour $C^+$ in $G^+$ containing $e_0$. By Lemma~\ref{lem:preprocess-opt}, $c^+(C^+)=\operatorname{OPT}_{\mathrm{RPP}}$.

We first bound the expected cost of the $Q$-join when $T\sim\mu_r$.

\begin{lemma}[Expected $Q$-join cost]\label{lem:expected-join}
Let $T\sim\mu_r$, let $Q=\operatorname{odd}(T)$, and let $J_T$ be a minimum-cost $Q$-join in $G^+$ with costs $c^+$. If $0<\beta\le\eta/4.1$ and $\eta\le10^{-12}$, then
\[
\mathbb E[c^+(J_T)]
\le
\frac{L+\operatorname{OPT}_{\mathrm{RPP}}}{4}
-\frac13\epsilon_{\mathrm P}\beta L
+218\eta\beta\operatorname{OPT}_{\mathrm{RPP}}.
\]
\end{lemma}

\begin{proof}
For the sole purpose of matching the formal hypotheses of
Lemma~\ref{lem:karlin-main}, consider the metric on the complete graph
with vertex set $V^+$ in which $e_0$ has length zero and every
other edge has length one. Every TSP tour containing $e_0$ then has
length $|V^+|-1$, whereas every TSP tour not containing $e_0$ has
length $|V^+|$. Hence $C^+$ is an optimum TSP tour for this
metric. This auxiliary choice of edge lengths is purely formal: the
slack statement of Lemma~\ref{lem:karlin-main} used below is
independent of these lengths, and any auxiliary metric making
$C^+$ an optimum tour would serve equally well.

Let $s$ and $s^*$ be the random functions given by Lemma~\ref{lem:karlin-main}. Define a vector $y$ for the $Q$-join polyhedron as follows. Put one unit on $e_0$. For every edge $e\in E^+\setminus\{e_0\}$ with $r_e>0$, put $r_e/4+s_e$. For every edge $e\in E(C^+)\setminus\{e_0\}$, add $1/4+s^*_e$. All remaining coordinates are zero. The vector $y$ is nonnegative because $s_e\ge-\beta r_e$, $\beta<1/4$, and $s^*\ge0$.

We show that $y$ is feasible for the $Q$-join polyhedron. Let $S\subseteq V^+$ satisfy $|S\cap Q|$ odd, or equivalently $|T\cap\delta(S)|$ odd. If $e_0\in\delta(S)$, the unit value on $e_0$ already satisfies the cut constraint. Otherwise, after replacing $S$ by its complement if necessary, both endpoints of $e_0$ lie outside $S$. Put $z^+=(r+\chi_{C^+})/2$, where $r$ is extended by zero to $e_0$.

If $z^+(\delta(S))\le2+\eta$, then Lemma~\ref{lem:karlin-main}(ii), together with $z^+(\delta(S))\ge2$, gives $y(\delta(S))\ge1$. Suppose instead that $z^+(\delta(S))>2+\eta$. Since $C^+$ is a tour, $\chi_{C^+}(\delta(S))\ge2$, and hence
\[
r(\delta(S))
=
2z^+(\delta(S))-\chi_{C^+}(\delta(S))
\le2\bigl(z^+(\delta(S))-1\bigr).
\]
Using $s_e\ge-\beta r_e$ and $s^*\ge0$, we obtain $y(\delta(S))\ge z^+(\delta(S))(1/2-2\beta)+2\beta$. The right-hand side is increasing in $z^+(\delta(S))$ for $\beta<1/4$. At $z^+(\delta(S))=2+\eta$, it equals $1+\eta/2-2\beta(1+\eta)$, which is at least one when $\eta\le10^{-12}$ and $\beta\le\eta/4.1$. Thus $y$ is feasible.

By the $Q$-join polyhedron in Section~\ref{subsec:polyhedra}, $c^+(J_T)\le c^+(y)$. Taking expectations and applying Lemma~\ref{lem:karlin-main}(iii), together with $c^+(r)=L$, $c^+(C^+)=\operatorname{OPT}_{\mathrm{RPP}}$, and $c^+_{e_0}=0$, proves the claim.
\end{proof}

The next lemma gives the feasibility property and the uniform $Q$-join bound needed for both the exact distribution $\mu_r$ and its approximation $\widehat\mu_r$.

\begin{lemma}\label{lem:feasibility-uniform}
Let $T$ be a spanning tree of $(V^+,E^+\setminus\{e_0\})$, let $Q=\operatorname{odd}(T)$, and let $J_T$ be a minimum-cost $Q$-join in $G^+$. If $T$ contains every edge of $\widehat R$, then the multiset union $T\uplus J_T$ can be transformed into a feasible RPP closed walk without increasing its cost. Moreover, for every spanning tree $T$ of $(V^+,E^+\setminus\{e_0\})$, $c^+(J_T)\le\operatorname{OPT}_{\mathrm{RPP}}/2$.
\end{lemma}

\begin{proof}
Assume first that $T$ contains every edge of $\widehat R$. The multigraph $T\uplus J_T$ is connected, and every vertex has even degree because $\operatorname{odd}(J_T)=Q=\operatorname{odd}(T)$. Identifying $a^r$ and $a^c$ undoes the split; $e_0$ then becomes a zero-cost loop, so its multiplicity can be set to zero. Replacing every required edge by the corresponding original required edge and every connector edge by its shortest path in $G$ gives a connected Eulerian multigraph containing every edge of $R$. An Eulerian closed walk in this multigraph is therefore feasible for the original RPP instance.

For the second assertion, $\chi_{C^+}/2$ is feasible for the $Q$-join polyhedron, since every nontrivial cut is crossed by the tour $C^+$ a positive even number of times and hence at least twice. Therefore $c^+(J_T)\le c^+(C^+)/2=\operatorname{OPT}_{\mathrm{RPP}}/2$.
\end{proof}

We first analyze the idealized version of the algorithm in which $T\sim\mu_r$.

\begin{lemma}\label{lem:exact-estimate}
Let $\eta=\epsilon_{\mathrm P}/(6\cdot218)$, $\beta=\eta/4.1$, and $\epsilon_0=\epsilon_{\mathrm P}\beta/6$. If $T\sim\mu_r$, the resulting RPP solution has expected cost at most $\left(\frac32-\epsilon_0\right)\operatorname{OPT}_{\mathrm{RPP}}$.
\end{lemma}

\begin{proof}
The chosen value of $\eta$ is smaller than $10^{-12}$. By the marginal identity preceding Lemma~\ref{lem:contract-required}, $\mathbb E[c^+(T)]=c^+(r)=L$. Combining this identity with Lemma~\ref{lem:expected-join} gives
\[
\mathbb E[c^+(T)+c^+(J_T)]
\le
\left(\frac54-\frac13\epsilon_{\mathrm P}\beta\right)L
+\left(\frac14+218\eta\beta\right)\operatorname{OPT}_{\mathrm{RPP}}.
\]
Since $L\le\operatorname{OPT}_{\mathrm{RPP}}$ by Lemma~\ref{lem:rpp-lp-lower} and $218\eta\beta=\epsilon_{\mathrm P}\beta/6$, the right-hand side is at most $(3/2-\epsilon_0)\operatorname{OPT}_{\mathrm{RPP}}$. By Lemma~\ref{lem:feasibility-uniform}, the resulting multigraph can be transformed into a feasible RPP closed walk without increasing its cost.
\end{proof}

\begin{proof}[Proof of Theorem~\ref{thm:main}]
The cases $|R|=0$ and $|R|=1$ are solved optimally by the first two steps of Algorithm~\ref{alg:rpp}. Assume $|R|\ge2$. The preprocessing and the optimization of \textup{\textsc{RPP-LP}} are polynomial time by Lemma~\ref{lem:rpp-lp-lower}. The splitting operation is explicit, a tree $T\sim\widehat\mu_r$ can be sampled in polynomial time by Lemma~\ref{lem:contract-required}, and a minimum-cost $Q$-join can be computed in polynomial time by the Edmonds--Johnson theorem. By Lemmas~\ref{lem:contract-required} and~\ref{lem:feasibility-uniform}, every sampled tree yields a feasible RPP closed walk.

It remains to compare the polynomial-time distribution $\widehat\mu_r$ with the exact maximum entropy distribution $\mu_r$. By Lemma~\ref{lem:contract-required}, $\mathbb E_{T\sim\widehat\mu_r}[c^+(T)]\le (1+\delta)c^+(r)=(1+\delta)L$, so relative to the exact-distribution estimate the tree-cost bound increases by at most $\delta L\le\delta\operatorname{OPT}_{\mathrm{RPP}}$.

For every spanning tree $T$, Lemma~\ref{lem:feasibility-uniform} gives $0\le c^+(J_T)\le\operatorname{OPT}_{\mathrm{RPP}}/2$. Since $\operatorname{TV}(\widehat\mu_r,\mu_r)\le\rho$, replacing $\mu_r$ by $\widehat\mu_r$ changes the expected $Q$-join cost by at most $\rho\operatorname{OPT}_{\mathrm{RPP}}/2$. Hence $\mathbb E[\textup{output cost}]\le\left(\frac32-\epsilon_0+\delta+\frac{\rho}{2}\right)\operatorname{OPT}_{\mathrm{RPP}}$. Choose $\delta\le\epsilon_0/4$ and $\rho\le\epsilon_0/2$. The approximation ratio is then at most $3/2-\epsilon_0/2$. Since $\frac{\epsilon_0}{2}=\frac{\epsilon_{\mathrm P}^2}{72\cdot218\cdot4.1}>10^{-36}$, Theorem~\ref{thm:main} follows.
\end{proof}

\begin{remark}[LP-relative and deterministic variants]
\label{rem:lp-relative-derandomization}
Let $\epsilon_{\mathrm{KKO}}>10^{-36}$ be the constant in the
LP-relative analysis of \citet{KarlinKleinOveisGharanFOCS2022}. We use
its tree-and-parity-correction bound before shortcutting: the coordinatewise
slack inequalities can be weighted by any nonnegative cost function.
Thus $c^+$ need not be metric, as is also used in the weighted
2-edge-connected multisubgraph consequence of that work.
By Lemma~\ref{lem:split-vector}, $r^0$ is feasible for
\textup{\textsc{TSP-LP}}, $r^0_{e_0}=1$, and
$c^+(r^0)=c^+(r)=L$. Since every edge of $\widehat R$ has
coordinate one, every tree in the support of $\mu_r$ contains
$\widehat R$; the proof of Lemma~\ref{lem:contract-required} identifies
this exact distribution with the lift of $\mu_{r'}$. Applying the
cited LP-relative tree-and-parity-correction analysis to $r^0$, and
using a minimum $Q$-join in place of its feasible parity-correction
vector, gives
\[
\mathbb E_{T\sim\mu_r}
 \bigl[c^+(T)+c^+(J_T)\bigr]
\le
\left(\frac32-\epsilon_{\mathrm{KKO}}\right)L.
\]

The approximate distribution used by the polynomial-time implementation
preserves this bound up to the prescribed accuracy. Indeed, for every
spanning tree $T$, the vector $r^0/2$ is feasible for the
$\operatorname{odd}(T)$-join polyhedron, because
$r^0(\delta(S))\ge2$ for every nontrivial cut. Hence
$c^+(J_T)\le L/2$, and Lemma~\ref{lem:contract-required} yields
\[
\mathbb E_{T\sim\widehat\mu_r}
 \bigl[c^+(T)+c^+(J_T)\bigr]
\le
\left(\frac32-\epsilon_{\mathrm{KKO}}
      +\delta+\frac\rho2\right)L.
\]
Fix $10^{-36}<\epsilon<\epsilon_{\mathrm{KKO}}$ and choose
$\delta,\rho>0$ with $\delta+\rho/2\le\epsilon_{\mathrm{KKO}}-\epsilon$.
This proves the LP-relative randomized $(3/2-\epsilon)$-guarantee.

For the deterministic variant, we compute the weighted-tree distribution and the conditional-expectation objective of \citet{KarlinKleinOveisGharanIPCO2023} on the uncontracted support of $r$, before fixing the required edges. Write $\nu$ for this $\lambda$-uniform distribution and $p_e:=\Pr_{T\sim\nu}[e\in T]$. We may require both the marginal accuracy prescribed in that work and $p_e\le(1+\tau)r_e$ for an arbitrarily small $\tau>0$. Indeed, the equal total marginals imply $|p_e-r_e|\le\tau(|V^+|-1)$; taking $\tau(|V^+|-1)\le\xi r_{\min}$ gives any prescribed relative accuracy $\xi$ on the positive support, where $r_{\min}:=\min\{r_e:r_e>0\}$. This requires only polynomially many bits by the determinant bound in Lemma~\ref{lem:contract-required}. The LP-relative structure is still defined from $r^0$ on $G^+$; $r'$ is not treated as a new solution of \textup{\textsc{TSP-LP}}.

Let $\overline y_T$ be the nonnegative vector from the cited parity-correction construction, averaged over its auxiliary Bernoulli variables after $T$ is fixed, with $(\overline y_T)_{e_0}=1$. We make one correction explicit when using approximate marginals. The root-cut argument in \citet{KarlinKleinOveisGharanFOCS2022} uses that both endpoints of $e_0$ are leaves under the exact distribution. In our support graph, $a^r$ is always a leaf, but this need not hold for $a^c$ under $\nu$. Set
\[
B_T:=\mathbf 1_{\{\deg_T(a^c)>1\}},
\qquad
y_T:=\overline y_T+B_T\frac{r^0}{2}.
\]
If $B_T=0$, both split vertices are leaves and the cited cut verification gives feasibility of $\overline y_T$. If $B_T=1$, the vector $r^0/2$ is feasible for the $\operatorname{odd}(T)$-join polyhedron, and $\overline y_T\ge0$, so $y_T$ is again feasible. Since every spanning tree has $\deg_T(a^c)\ge1$ and $r(\delta(a^c))=1$, we have
\[
\Pr_{T\sim\nu}[B_T=1]
\le\mathbb E_{T\sim\nu}[\deg_T(a^c)]-1
=p(\delta(a^c))-1\le\tau.
\]
The added expected cost is therefore at most $\tau L/2$. Taking $\tau$ sufficiently small absorbs this additional error, together with the marginal-approximation error in the cited analysis, while retaining a constant $\epsilon_*>10^{-36}$. Thus, with $Z(T):=c^+(T)+c^+(y_T)$,
\[
0\le c^+(T)+c^+(J_T)\le Z(T),
\qquad
\mathbb E_{T\sim\nu}[Z(T)]\le\left(\frac32-\epsilon_*\right)L.
\]
Here we use the cited expected bound on the auxiliary objective itself, rather than infer it from the expected cost of a minimum $Q$-join. All distribution-dependent coefficients are computed for $\nu$ before conditioning and are then kept fixed.

Let $\mathcal F:=\{T:\widehat R\subseteq T\}$. Since $r_e=1$ for $e\in\widehat R$ and $p$ and $r$ have the same total value,
\[
\Pr_{T\sim\nu}[\mathcal F^c]
\le\sum_{e\in\widehat R}(1-p_e)
\le\sum_{e:p_e<r_e}(r_e-p_e)
\le\tau(|V^+|-1).
\]
Thus we can make $\nu(\mathcal F)\ge1-\gamma$ for any fixed $\gamma\in(0,1)$. Nonnegativity of $Z$ gives the initial conditional bound
\[
\Phi(\mathcal S_0):=\mathbb E_{T\sim\nu}[Z(T)\mid\mathcal F]
\le\frac{\mathbb E_{T\sim\nu}[Z(T)]}{\nu(\mathcal F)}
\le\frac{\frac32-\epsilon_*}{1-\gamma}\,L,
\]
where $\mathcal S_0$ fixes every required edge to one. Unlike the exact distribution $\mu_r$, $\nu$ need not contain these edges with probability one; the displayed inequality accounts for the conditioning loss.

We now continue the cited conditional-expectation procedure from $\mathcal S_0$, processing only nonrequired edges. The distribution $\nu\mid\mathcal F$ is represented by contracting $\widehat R$, retaining each remaining edge with its weight $\lambda_e$ and keeping parallel edges distinct. Subsequent inclusions are contracted and exclusions are deleted. All parity events and the objective remain on the lifted tree in $G^+$. In particular, for $A\subseteq E^+\setminus\{e_0\}$ and its corresponding nonrequired edge set $A'$ in $G'$, we have
\[
|A\cap T|=|A\cap\widehat R|+|A'\cap T'|.
\]
Hence fixed edges contribute only deterministic offsets to the cardinality and parity conditions used by the cited procedure. Its matrix-tree and generating-polynomial evaluations therefore compute the required conditional objectives on the contracted representation, with the auxiliary Bernoulli variables already averaged. The added term contributes $(L/2)\Pr[B_T=1\mid\mathcal S]$, which is computed by the same cardinality oracle as $\Pr[B_T=1\mid\mathcal S]=1-\Pr[\deg_T(a^c)=1\mid\mathcal S]$.

With exact evaluations, choosing a positive-probability branch of minimum conditional value never increases $\Phi$. Polynomial-precision evaluations suffice as well: if there are at most $M:=|E'|$ decisions and each branch value is computed to additive error at most $\zeta L/(2M)$, the total possible increase is at most $\zeta L$. Choose fixed $\gamma,\zeta>0$ and $\epsilon_{\mathrm{det}}>10^{-36}$ such that
\[
\frac{\frac32-\epsilon_*}{1-\gamma}+\zeta
\le\frac32-\epsilon_{\mathrm{det}}.
\]
Such choices exist because $\epsilon_*>10^{-36}$. The resulting tree $T\supseteq\widehat R$ then satisfies
\[
c^+(T)+c^+(J_T)\le Z(T)
\le\Phi(\mathcal S_0)+\zeta L
\le\left(\frac32-\epsilon_{\mathrm{det}}\right)L.
\]
If $L=0$, the conclusion instead follows directly by choosing a spanning tree in the support of $r$ containing $\widehat R$ and using $c^+(J_T)\le L/2=0$.

The cost comparison uses the nonnegative-cost tree-and-parity-correction analysis underlying the weighted 2-edge-connected multisubgraph result, not the final metric shortcutting step. By Lemma~\ref{lem:feasibility-uniform}, the resulting $T\uplus J_T$ expands into a feasible RPP closed walk without increasing its cost. Thus an LP-relative better-than-$3/2$ guarantee is attained deterministically in polynomial time.
\end{remark}

\section{A general reduction from RPP to metric TSP}\label{sec:generic-reduction}

The preprocessing in Section~\ref{subsec:preprocessed-graph} expresses the original RPP instance on $G$ with required edge set $R$ as the problem of finding a minimum-cost tour in $\widehat G$ that traverses every edge of $\widehat R$. We now use the transformation described by \citet[Section~2.1]{Lampis2014} to remove the requirement that these edges must be traversed. We give the details needed here and later show that the same transformation also preserves approximation guarantees measured against the LP relaxation defined on the preprocessed RPP instance.

Let $c(R):=\sum_{e\in R}c_e=\sum_{i=1}^m\widehat c_{\widehat e_i}$, and fix an integer $k\ge2$. For each required edge $\widehat e_i=i^-i^+$, remove $\widehat e_i$ and replace it by a path $P_i:\ p_{i,0}=i^-,p_{i,1},\ldots,p_{i,k}=i^+$, where every edge $p_{i,j-1}p_{i,j}$ has cost $\widehat c_{\widehat e_i}/k$. Each new internal vertex of $P_i$ has degree two. Keep every connector edge of $\widehat G$ with its original cost. Let $H_k=(V_k,E_k)$ be the resulting graph and let $c^k$ be its edge cost function.

For $u,v\in V_k$, let $d_k(u,v)$ be the shortest-path distance from $u$ to $v$ in $H_k$ with respect to $c^k$. The complete graph on $V_k$ with edge costs $d_k$ is a metric TSP instance, which we denote by $I_k$. The distances $d_k$ can be computed in polynomial time.

\begin{lemma}\label{lem:subdivision-opt-upper}
The optimum tour cost of $I_k$ is at most $\operatorname{OPT}_{\mathrm{RPP}}$ for the original RPP instance on $G$ with required edge set $R$.
\end{lemma}

\begin{proof}
By Lemma~\ref{lem:preprocess-opt}, a minimum-cost tour in $\widehat G$ that visits every vertex exactly once and traverses every edge of $\widehat R$ has cost $\operatorname{OPT}_{\mathrm{RPP}}$. Replacing each required edge $\widehat e_i$ of such a tour by the entire path $P_i$ gives a tour on $V_k$ of the same cost with respect to $c^k$. Replacing its edges by the corresponding distances $d_k$ cannot increase the cost. Hence $\operatorname{OPT}_{\mathrm{TSP}}(I_k)\le \operatorname{OPT}_{\mathrm{RPP}}$.
\end{proof}

\begin{lemma}\label{lem:lift-metric-tour}
From any feasible tour $C$ of metric TSP instance $I_k$, one can construct in polynomial time a feasible closed walk $W$ for the original RPP instance on $G$ with required edge set $R$ such that $c(W)\le d_k(C)+\frac{2}{k}c(R)$.
\end{lemma}

\begin{proof}
Replace every edge of $C$ by a shortest path in $H_k$ with respect to $c^k$. Their cyclic concatenation is a closed walk that visits every vertex of $V_k$. Let $M$ be its edge multiset. Then $(V_k,M)$ is connected and Eulerian, and $c^k(M)=d_k(C)$.

Fix $i\in\{1,\ldots,m\}$. At most one edge of path $P_i$ has multiplicity zero in $M$. Indeed, if two distinct edges of $P_i$ had multiplicity zero, the vertices between them would contain a nonempty component of $(V_k,M)$ consisting only of internal vertices of $P_i$. Since every such vertex has degree two in $H_k$, this component would have no edge to the remainder of the multigraph, contradicting connectedness.

If one edge of $P_i$ is missing, increase its multiplicity by two. Performing this operation for every $i$ preserves connectedness and all degree parities and increases the cost by at most $\frac{2}{k}\sum_{i=1}^m\widehat c_{\widehat e_i}=\frac{2}{k}c(R)$. Thus the resulting multigraph is connected and Eulerian and contains every edge of every path $P_i$.

It remains to replace each $P_i$ by the corresponding required edge $\widehat e_i$ with the appropriate multiplicity. Let $m_{i,j}>0$ be the multiplicity of $p_{i,j-1}p_{i,j}$. Since every internal vertex $p_{i,j}$ has even degree, $m_{i,j}+m_{i,j+1}\equiv0\pmod 2$, so all edge multiplicities along $P_i$ have the same parity. If they are odd, replace every $m_{i,j}$ by one; if they are even, replace every $m_{i,j}$ by two. This does not increase the cost and preserves the parity of the degree of each endpoint. The path $P_i$ can then be replaced by $\widehat e_i$ with multiplicity one or two, respectively, with the same cost.

After applying this operation to every $i$, we obtain a connected Eulerian multigraph in $\widehat G$ containing every edge of $\widehat R$. Replace each edge of $\widehat R$ by its corresponding required edge in $R$, and replace each connector edge $uv$ by a shortest path in $G$ from $\pi(u)$ to $\pi(v)$. The resulting closed walk is feasible for the original RPP instance and satisfies the claimed bound.
\end{proof}

\begin{proof}[Proof of Theorem~\ref{thm:generic-transfer}]
The cases $|R|\le1$ are solved optimally as in Algorithm~\ref{alg:rpp}. Assume $|R|\ge2$, and set $k=\max\{2,\lceil2/\varepsilon\rceil\}$. Run the $\alpha$-approximation algorithm for metric TSP on $I_k$, obtaining a tour $C$, and lift $C$ to a feasible RPP closed walk $W$ by Lemma~\ref{lem:lift-metric-tour}. Then
\[
\begin{aligned}
c(W)
&\le
d_k(C)+\frac{2}{k}c(R)\\
&\le
\alpha\,\operatorname{OPT}_{\mathrm{TSP}}(I_k)
+\frac{2}{k}c(R)\\
&\le
\left(\alpha+\frac{2}{k}\right)
\operatorname{OPT}_{\mathrm{RPP}}\\
&\le
(\alpha+\varepsilon)\operatorname{OPT}_{\mathrm{RPP}},
\end{aligned}
\]
where the third inequality follows from
Lemma~\ref{lem:subdivision-opt-upper} and from
$c(R)\le\operatorname{OPT}_{\mathrm{RPP}}$, since every feasible RPP
closed walk traverses every edge of $R$ at least once and all edge
costs are nonnegative. Since $H_k$ has $O(k|R|)$ vertices, and
$k$ is a constant for every fixed $\varepsilon>0$, the entire
construction runs in polynomial time.
\end{proof}

The LP-relative form of this transfer and the corresponding integrality-gap relation are recorded in Appendix~\ref{app:lp-consequences}.

\section{Conclusion}\label{sec:conclusion}

We considered approximation algorithms for RPP. We showed that the
technique of sampling from maximum entropy distributions for metric TSP
can be adapted to RPP, yielding a randomized polynomial-time
$(3/2-\epsilon)$-approximation algorithm for some absolute constant
$\epsilon>10^{-36}$. The analysis is also relative to
\textup{\textsc{RPP-LP}}, and the corresponding guarantee can be
attained deterministically using the conditional-expectation method of
\citet{KarlinKleinOveisGharanIPCO2023}. We also described how an
$\alpha$-approximation algorithm for metric TSP yields, for every fixed
$\varepsilon>0$, an $(\alpha+\varepsilon)$-approximation algorithm
for RPP.

Finally, the same guarantees apply to the GRP. RPP is a special case of GRP. Conversely, a GRP instance can be reduced cost-preservingly to RPP by, for each required vertex $v$ not already incident to a required edge, introducing a fresh leaf $v'$ and the zero-cost required edge $vv'$. Traversing $vv'$ forces the resulting RPP walk to visit $v$, while any feasible GRP walk that visits $v$ can traverse the zero-cost detour $vv'v$. Thus the optimum value is preserved, and both algorithms extend immediately to GRP.

\paragraph{Acknowledgment.}
Generative AI tools were used for language editing and minor editorial assistance. The authors take full responsibility for the manuscript.

\appendix
\section{LP-relative consequences}\label{app:lp-consequences}

The subdivision transformation is also compatible with the LP relaxations.

\begin{lemma}\label{lem:subdivision-lp}
Let $x$ be feasible for \textup{\textsc{RPP-LP}}. Define $z^k$ on
the metric TSP instance $I_k$ by retaining $x_e$ on every connector
edge, assigning value one to every edge of every path $P_i$, and
assigning zero to all remaining edges. Then $z^k$ is feasible for
\textup{\textsc{TSP-LP}} and
\[
d_k(z^k)
\le
\sum_{e\in\widehat E_{\mathrm{con}}}\widehat c_e x_e
+\sum_{e\in\widehat R}\widehat c_e.
\]
Consequently, $L_{\mathrm{TSP}}(I_k)\le L$.
\end{lemma}

\begin{proof}
The degree equations are immediate. Consider a nontrivial cut
$\delta(S)$ of $H_k$. If some path $P_i$ crosses the cut at least
twice, then $z^k(\delta(S))\ge2$. Otherwise, every $P_i$ crosses at
most once. In this case $U:=S\cap\widehat V$ is a nonempty proper
subset of $\widehat V$; otherwise, since $S$ is nontrivial, some
$P_i$ would cross the cut at least twice. Moreover,
\[
z^k(\delta_{H_k}(S))
=
\overline x(\delta_{\widehat G}(U))
\ge2,
\]
where the equality follows because a path $P_i$ crosses
$\delta(S)$ exactly once precisely when its endpoints are separated
by $U$. Thus $z^k$ satisfies all subtour cut constraints.

The cost inequality follows directly from the construction of $H_k$
and its metric completion $I_k$. Taking $x=x^*$ gives
$L_{\mathrm{TSP}}(I_k)\le L$.
\end{proof}

\begin{corollary}[LP-relative transfer]\label{cor:generic-lp-transfer}
Suppose that, for every metric TSP instance $I$, an algorithm returns
a tour of cost at most $\alpha L_{\mathrm{TSP}}(I)$. Then, for every
fixed $\varepsilon>0$, the transformation above yields a feasible RPP
closed walk $W$ satisfying $c(W)\le(\alpha+\varepsilon)L$, where $L$ is the optimum value of \textup{\textsc{RPP-LP}} on the
preprocessed RPP instance.
\end{corollary}

\begin{proof}
Set $k=\max\{2,\lceil2/\varepsilon\rceil\}$. Repeat the proof of Theorem~\ref{thm:generic-transfer}, using
Lemma~\ref{lem:subdivision-lp} in place of
Lemma~\ref{lem:subdivision-opt-upper}. Since every required edge has
value one in the extension of every feasible solution of
\textup{\textsc{RPP-LP}} and all edge costs are nonnegative,
$c(R)\le L$. Therefore Lemma~\ref{lem:lift-metric-tour} gives
\[
c(W)
\le
\alpha L_{\mathrm{TSP}}(I_k)+\frac{2}{k}c(R)
\le
\left(\alpha+\frac{2}{k}\right)L
\le
(\alpha+\varepsilon)L.
\]
\end{proof}

\subsection{Equality of the integrality gaps}\label{sec:gap}

For an RPP instance $J$ with at least two required edges, let
$L_{\mathrm{RPP}}(J)$ denote the optimum value of \textup{\textsc{RPP-LP}}
on its preprocessed instance. Let $\gamma_{\mathrm{TSP}}$ and
$\gamma_{\mathrm{RPP}}$ denote the corresponding worst-case
integrality gaps.

\begin{theorem}\label{thm:gap-equality}
The two relaxations have the same worst-case integrality gap: $\gamma_{\mathrm{RPP}}=\gamma_{\mathrm{TSP}}$. Consequently, the worst-case integrality gap of \textup{\textsc{RPP-LP}} on preprocessed RPP instances is strictly below $3/2$ by \citet{KarlinKleinOveisGharanFOCS2022}.
\end{theorem}

\begin{proof}
We first show $\gamma_{\mathrm{RPP}}\le\gamma_{\mathrm{TSP}}$.
Fix an RPP instance $J$ with cost function $c_J$ and required edge set $R_J$, let $L=L_{\mathrm{RPP}}(J)$, and apply
the subdivision transformation with parameter $k$. By
Lemmas~\ref{lem:lift-metric-tour} and~\ref{lem:subdivision-lp}, $\operatorname{OPT}_{\mathrm{RPP}}(J)\le\gamma_{\mathrm{TSP}}L+\frac{2}{k}c_J(R_J)$. Letting $k\to\infty$ gives
$\operatorname{OPT}_{\mathrm{RPP}}(J)\le\gamma_{\mathrm{TSP}}L$,
and hence $\gamma_{\mathrm{RPP}}\le\gamma_{\mathrm{TSP}}$.

For the reverse inequality, let $I=(V,d)$ be a metric TSP instance.
For each $v\in V$, create two vertices $v^-$ and $v^+$ joined
by a zero-cost required edge. For every distinct $u,v\in V$, add
the four connector edges between the two corresponding pairs, each
of cost $d_{uv}$. Denote the resulting RPP instance by $J(I)$. Because the required edges form a matching covering all vertices and the connector costs are induced by the metric $d$, this instance is already in the preprocessed form of Section~\ref{subsec:preprocessed-graph}, up to renaming the private vertices.

Contracting the zero-cost required edges maps every feasible RPP
solution to a closed walk visiting all vertices of $I$, which can
be shortcut to a TSP tour without increasing its cost. Conversely,
every TSP tour can be expanded by inserting the zero-cost required
edge corresponding to each vertex. Thus
$\operatorname{OPT}_{\mathrm{RPP}}(J(I))
 =\operatorname{OPT}_{\mathrm{TSP}}(I)$.

Given a feasible
\textup{\textsc{RPP-LP}} solution $\widehat x$, define
$x_{uv}=\sum_{\sigma,\tau\in\{-,+\}}
\widehat x_{u^\sigma v^\tau}$. The degree constraints project to
$x(\delta(v))=2$, while applying an RPP cut constraint to the union
of both copies of every vertex in $U\subsetneq V$ gives
$x(\delta(U))\ge2$. Hence $x$ is a cost-preserving feasible
solution of \textup{\textsc{TSP-LP}}.

Conversely, from a feasible \textup{\textsc{TSP-LP}} solution $x$,
set $\widehat x_{u^\sigma v^\tau}=x_{uv}/4$. Every copy then has
connector degree one. For any cut, if at least two required pairs are
split, the required edges already contribute at least two. If no pair
is split, the connector contribution is the corresponding TSP cut
value. If exactly one pair, say that of $v$, is split, and $U$
denotes the vertices whose two copies lie on the same side as the
selected copy of $v$, the connector contribution is
$\frac12x(\delta(U))+\frac12x(\delta(U\cup\{v\}))\ge1$. If both sets define nontrivial cuts, this follows from the cut constraints of \textup{\textsc{TSP-LP}}; in the boundary cases $U=\emptyset$ or $U=V\setminus\{v\}$, it follows from the degree equation $x(\delta(v))=2$. Together with the crossing required edge, the cut has value at least two. Thus the lift is feasible and cost preserving.

Therefore
$L_{\mathrm{RPP}}(J(I))=L_{\mathrm{TSP}}(I)$, and each metric TSP
integrality-gap ratio is reproduced exactly by an RPP instance.
Hence $\gamma_{\mathrm{TSP}}\le\gamma_{\mathrm{RPP}}$, completing
the proof.
\end{proof}

\begingroup\small
\setlength{\bibsep}{3pt}

\endgroup

\end{document}